\documentclass[runningheads]{llncs}
\usepackage{graphicx}

\graphicspath{{figures/}}

\usepackage[noend,ruled,linesnumbered]{algorithm2e}
\SetKwRepeat{Do}{do}{while}
\usepackage{amsmath}
\usepackage{amssymb}
\usepackage{booktabs}
\usepackage[utf8]{inputenc}
\usepackage{cite}
\usepackage{hyperref}
\hypersetup{
hidelinks,
colorlinks=true,
citecolor=[rgb]{0.121 0.47 0.705},
linkcolor=[rgb]{0.121 0.47 0.705},
urlcolor=[rgb]{0.121 0.47 0.705}
}
\usepackage[basic]{complexity}
\usepackage{enumerate}
\usepackage{pifont}
\usepackage{refcount}
\usepackage{wasysym}
\usepackage[dvipsnames]{xcolor}
\usepackage{xspace}
\usepackage{doi,url}

\spnewtheorem{observation}{Observation}{\bfseries}{\itshape}

\usepackage{thm-restate} 

\newcommand*{\mysc}[1]{{\normalfont\textsc{#1}}\xspace}
\newcommand*{\EVEN}{\mysc{even}}
\newcommand*{\LOW}{\mysc{low}}
\newcommand*{\MID}{\mysc{mid}}
\newcommand*{\HIGH}{\mysc{high}}

\newcommand*{\Pu}{P}
\newcommand*{\VA}{V_{\!A}}

\newcommand{\etal}{et~al.\xspace}

\makeatletter
\newcommand\nextitem[1]{%
  \setcounter{\@enumctr}{#1}%
  \addtocounter{\@enumctr}{-1}%
}
\newcommand\nextitemref[1]{%
  \setcounterref{\@enumctr}{#1}%
  \addtocounter{\@enumctr}{-1}%
}
\makeatother

\renewcommand{\subsubsection}[1]{\par\medskip\noindent\textbf{#1.}\enskip}

\begin{document}

\title{Short Plane Supports for Spatial Hypergraphs}
\titlerunning{Short Plane Supports for Spatial Hypergraphs}

\author{%
    Thom~Castermans\inst{1}
    \and
    Mereke~van~Garderen\inst{2}
    \and
    Wouter~Meulemans\inst{1}
    \and
    Martin~N\"ollenburg\inst{3}
    \and
    Xiaoru~Yuan\inst{4}}

\authorrunning{T. Castermans et al.}
\institute{%
    TU Eindhoven, the Netherlands
    \email{[t.h.a.castermans, w.meulemans]@tue.nl} \and
    Universit\"at Konstanz, Germany
    \email{mereke.van.garderen@uni-konstanz.de} \and
    TU Wien, Vienna, Austria
    \email{noellenburg@ac.tuwien.ac.at} \and
    Peking University, Beijing, China
    \email{xiaoru.yuan@pku.edu.cn}}

\maketitle 

\begin{abstract}
A graph $G=(V,E)$ is a \emph{support} of a hypergraph $H=(V,S)$ if every hyperedge induces a connected subgraph in $G$. Supports are used for certain types of hypergraph visualizations.
In this paper we consider visualizing \emph{spatial} hypergraphs, where each vertex has a fixed location in the plane.
This is the case, e.g., when modeling set systems of geospatial locations as hypergraphs.
By applying established aesthetic quality criteria we are interested in finding supports that yield plane straight-line drawings with minimum total edge length on the input point set~$V$.
We first show, from a theoretical point of view, that the problem is \NP-hard already under rather mild conditions as well as a negative approximability results.
Therefore, the main focus of the paper lies on practical heuristic algorithms as well as an exact, ILP-based approach for computing short plane supports.
We report results from computational experiments that investigate the effect of requiring planarity and acyclicity on the resulting support length.
Further, we evaluate the performance and trade-offs between solution quality and speed of several heuristics relative to each other and compared to optimal solutions.

\end{abstract}

\section{Introduction}
\label{sec:intro}
A \emph{hypergraph} $H = (V,S)$ is a generalization of a graph, in which each hyperedge in $S$ is a nonempty subset of the vertex set $V$, that is, $S \subseteq \mathcal{P}(V) \setminus \{ \emptyset \}$.
Furthermore, we assume here that every element $v \in V$ is in at least one hyperedge $s \in S$.
Hypergraphs arise in many domains to model set systems representing clusters, groups or other aggregations.
To allow for effective exploration and analysis of such data, visualization is often used.
Indeed, drawing hypergraphs relates to set visualization, an active subfield of information visualization (see the recent survey of Alsallakh~\etal~\cite{Alsallakh2016}).
Various methods have been developed to visualize set systems for elements fixed in (geo)spatial positions, such as Bubble Sets \cite{bubblesets}, LineSets \cite{Alper2011}, Kelp Diagrams \cite{Dinkla2012} and Kelp Fusion \cite{Meulemans2013}.
These methods make different trade-offs between, e.g., Gestalt theory and Tufte's principle of ink minimization \cite{tufte} to visually convey the set structures; user studies have been performed to analyze the effectiveness of such trade-offs \cite{Meulemans2013}.

\begin{figure}[t]
  \includegraphics[page=1]{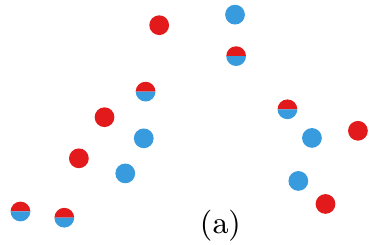}
  \hfill
  \includegraphics[page=2]{problem}
  \hfill
  \includegraphics[page=3]{problem}
  \vspace{-\baselineskip}
  \caption{(a) A set system with colors indicating set membership. (b) The shortest plane support of the corresponding hypergraph. (c) A Kelp-style rendering of the set system.}
  \label{fig:problem}
\end{figure}

An important concept to model the drawing of hypergraphs is that of a hypergraph support \cite{supportfirst}:
a \emph{support} of a hypergraph $H = (V,S)$ is a graph~$G = (V,E)$ such that every hyperedge $s \in S$ induces a connected subgraph in~$G$.
In other words, for every hyperedge $s$, the restriction of $G$ to only edges that connect vertices in $s$, denoted $G[s]$, is connected and spans all vertices in $s$.
Hypergraph supports correspond to a prominent visualization style for geospatial sets, namely that of connecting all elements of a set using colored links, such as seen in Kelp-style diagrams \cite{Dinkla2012,Meulemans2013} (see also Fig.~\ref{fig:problem}) or LineSets \cite{Alper2011}.
Thus, finding an embedded support that satisfies certain criteria readily translates into a good rendering of the spatial set system.
A ``good'' support should avoid edge crossings, a standard quality criterion in the graph-drawing literature \cite{Purchase2002}.
Moreover, as per Tufte's principle of ink minimization \cite{tufte}, it should have small total edge length.
Of course, one may argue that edges of the support that are used by multiple hyperedges do not significantly reduce the ``ink'' and thus multiplicity should be considered.
However, we observe that such edges show co-occurrences of elements and thus have a potential added value in the drawing---user studies that establish the validity of this reasoning are beyond the scope of this paper.
The shortest support need not be a tree, but to further build on this idea of co-occurrences, one may want to restrict the support to be acyclic---a support tree.

In many applications, the vertices have some associated (geo)spatial location, thereby prescribing their positions in the drawing of the support.
We focus on this case where vertices have fixed positions in the plane and study supports that are embedded using straight-line edges.
Fig.~\ref{fig:realdata} shows an example on real-world data of restaurants, similar to those used in \cite{Meulemans2013}.

\subsubsection{Contributions}
The contributions of this paper are two-fold: on the one hand we fill some gaps in theoretical knowledge about computing plane supports and support trees; on the other hand, we perform computational experiments to gain more insight into the trade-offs on the complexity of the visual artifact for (implicit) support-based set visualization methods. Our focus is on the latter.

In Section~\ref{sec:theory} we explore computational aspects of the problem and introduce our algorithms.
We observe that plane support trees always exist if at least one vertex is contained in all hyperedges, but show that length minimization is \NP-hard.
Moreover, the natural approach to extend a minimum spanning tree does not even yield a constant-factor approximation.
Finally, we present two heuristics, one based on local search, the other on iteratively computing minimum spanning trees, as well as an exact integer linear program (ILP).

\begin{figure}[t]
  \includegraphics[page=1]{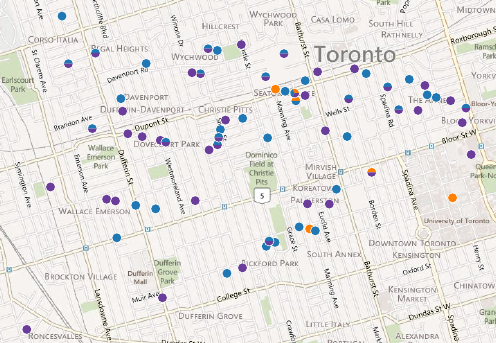}
  \hfill
  \includegraphics[page=2]{toronto_filtered_d-dd_4-5_japanese}
  \hfill
  \includegraphics[page=4]{toronto_filtered_d-dd_4-5_japanese}
  \vspace{-\baselineskip}
  \caption{A set system of restaurants in downtown Toronto: input memberships and locations (left) and a Kelp-style rendering of the shortest plane support (right).}
  \label{fig:realdata}
  \vspace{-\baselineskip}
\end{figure}

In Section~\ref{sec:experiments} we describe the results of two computational experiments.
The first experiment compares the performance of the two heuristic algorithms in terms of quality and speed.
Whereas the local search achieves better quality, the approximation algorithm is faster.
The second experiment compares how well these algorithms perform compared to the optimum, computed via the ILP, and investigates the cost in terms of edge length incurred by requiring planarity or acyclicity.
The effect of planarity and acyclicity seems to be predictably influenced by the number of hyperedges and the number of incident hyperedges per vertex, but not by the number of vertices.
Moreover, the experiment shows that local search often achieves an optimal result.

\subsubsection{Related work}
Regarding supports for elements with fixed locations, some results are already known.
The results of Bereg~\etal~\cite{redbluetrees} imply that existence of a plane support tree for two disjoint hyperedges can be tested in polynomial time; this implies the same result for a plane support.
This problem has also been studied in a setting with additional Steiner points~\cite{steiner,efrat2014mapsets}.
Van Goethem~\etal~\cite{paintersarxiv} enforce a stricter planarity than that of planar supports and investigate the resulting properties for elements on a regular grid, where only neighboring elements can be connected. However, solution length is of no concern in their results.

Without the planarity requirement, existence and length minimization of a (nonplane) support tree for fixed elements can be solved in polynomial time~\cite{treemin,kmn-mtshled-14}.
Hurtado~\etal~\cite{coloredspanninggraphs} show that length minimization of a support for two hyperedges is solvable in polynomial time.
However, for three or more hyperedges this problem is \NP-hard~\cite{multicolor}.
We show that this is in fact hard for two hyperedges if we do require planarity.

Planar supports without fixed elements have also received attention.
Johnson and Pollak~\cite{supportfirst} originally showed that deciding whether a planar support exists is \NP-hard; various restrictions have since been proven to be \NP-hard (e.g., \cite{planarsupports}). Contrasting these reductions, our hardness result (Theorem~\ref{thm:nphard}) requires only two hyperedges, but uses length minimization.
Buchin~\etal~\cite{planarsupports} show that testing for a planar support tree with bounded maximum degree is solvable in polynomial time; testing for a planar support tree such that the induced subgraph of each hyperedge is Hamiltonian can also be done in polynomial time~\cite{brandes}.

Various set-visualization methods \cite{Alper2011,Dinkla2012,Meulemans2013} implicitly also compute supports, considering various criteria such as length, detour, shape, crossings, and bends.

\section{Computing short plane supports}
\label{sec:computation}

We first describe our theoretical results. Omitted proofs are in Appendix~\ref{app:proofs}.

\label{sec:theory}
\subsubsection{Existence}
The observation below gives a sufficient condition for the existence of a plane support tree. Bereg~\etal~\cite{redbluetrees} provide a necessary condition for $|S| = 2$, though the problem remains open for $|S| > 2$.

\begin{observation}\label{obs:star}
	Consider a hypergraph $H=(V,S)$ with no three vertices in $V$ on a line, such that $\VA = \bigcap_{s \in S} s \neq \emptyset$. Then $H$ has a plane support tree.
\end{observation}
\begin{proof}
We use the Euclidean minimum spanning tree on $\VA$ and connect each vertex in $V\setminus \VA$ to the closest one in $\VA$.
This readily yields a support tree; it is plane as no crossings are created when connecting to the closest point in $\VA$ and no overlaps are created in the absence of collinear points.
\qed
\end{proof}

Without a vertex in $\VA$, one can immediately construct instances that enforce a crossing in any support, e.g., an X-configuration of two disjoint hyperedges.

\subsubsection{Approximation}
In a support tree the subgraph induced by $\VA$ must be a connected subtree to satisfy the support property for all hyperedges.
Next we consider  using the above idea to start with an Euclidean minimum spanning tree (EMST) of $\VA$ and extend it to a support tree.
Though this leads to an approximation algorithm for two hyperedges~\cite{coloredspanninggraphs} if we allow intersections, we show below that the planarity requirement can cause the resulting support length to exceed any constant factor of the length of the shortest plane support tree.

\begin{restatable}{lemma}{lemapprox}
    \label{lem:approx}
    There is a family of $n$-vertex hypergraphs $H=(V, \{r,b\})$ with $\VA = r \cap b \ne \emptyset$ such that any plane support of $H$ that includes an EMST of $\VA$ is a factor $\Theta(|V|)$ longer than the shortest plane support tree.
\end{restatable}

\begin{proof}[sketch]
The family is drawn in Fig.~\ref{fig:approx}. The convex chains force the support with length $\Theta(n) \cdot \ell$ when the EMST on $\VA$ is used. Using a different tree on $\VA$ can give a total of length $\Theta(1) \cdot \ell$.
\qed
\end{proof}

\begin{figure}[b]
  \centering
  \includegraphics{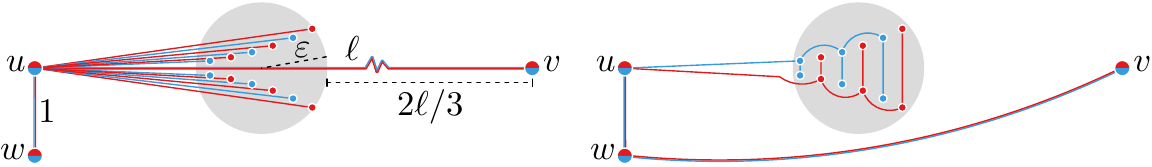}
  \vspace{-0.8\baselineskip}
  \caption{An $n$-point instance with approximation ratio $\Theta(n)$ if using an EMST on $\VA$. All edges are straight-line segments; curvature emphasizes the effect of the convex chain.}
  \label{fig:approx}
\end{figure}

Removing vertex $w$ from construction in Fig.~\ref{fig:approx}, we can similarly show that a plane support tree, which now necessarily includes the edge $uv$, is a factor $\Theta(n)$ longer than a shortest nonplane support tree. 

\begin{corollary}\label{cor:plane}
	There is a family of $n$-vertex hypergraphs $H=(V,\{r,b\})$ with $\VA = r \cap b \ne \emptyset$ such that any plane support tree of $H$ is a factor $\Theta(n)$ longer than the shortest nonplane support tree.
\end{corollary}

\label{sec:hardness}
\subsubsection{Computational complexity}
Unfortunately, finding the shortest plane support and several restricted variants are \NP-hard, as captured in the theorem below.
It uses a fairly straightforward reduction from planar monotone 3-SAT~\cite{l-pftu-82}.

\begin{restatable}{theorem}{thmnphard}
\label{thm:nphard}
Let $H = (V,\{r, b\})$ be a hypergraph with vertices $V$ having fixed locations in~$\mathbb{R}^2$ and with $r \subseteq b$ or $r \cap b =\emptyset$. It is \NP-hard to decide whether~$H$ admits a plane support tree with length at most $L$ for some $L > 0$.
\end{restatable}

\subsection{Iterative minimum spanning trees}
\label{sec:mstiteration}
Here we focus on computing short supports without requiring planarity.
As described by Hurtado et al.~\cite{coloredspanninggraphs}, EMSTs can be used to find an approximation of the shortest support.
In particular, let $H = (V,S)$ be a hypergraph with $n$ vertices and $k$ hyperedges; by computing an EMST for each hyperedge and taking their union, we get a support that is a $k$-approximation\footnote{One can actually do slightly better, by computing spanning trees on the intersection of two hyperedges, yielding roughly a $(0.8k)$-approximation \cite{coloredspanninggraphs}.} of the shortest support. This algorithm runs in $O(k n \log n)$ time.

Suppose that we compute the EMSTs $T_1, \ldots, T_k$ in that order, for the $k$ hyperedges in $S$.
The final support is  the union of these trees: its length is not increased by using an edge in $T_i$ that is already present in some $T_j$ ($j < i$).
Hence, we can consider any pair of vertices that is adjacent in $T_1 \cup \ldots \cup T_{i-1}$ to have distance zero, when computing $T_i$.
This heuristically reduces the length of the resulting support (though the approximation ratio remains the same).
However, the order in which hyperedges are considered now matters for the result. To alleviate this issue, we iteratively recompute the minimum spanning trees.

\subsubsection{Algorithm} 
We define a \emph{computation sequence} $\sigma$ of a hypergraph $H = (V,S)$ as a sequence of hyperedges that contains each hyperedge in $S$ at least once.
Each item $s$ in the sequence $\sigma$ represents the computation of the (not-quite Euclidean) MST on the vertices of $s$, such that distances between pairs of vertices that are part of the current support have weight $0$ and weight equal to their Euclidean distance otherwise.
We use $T_s$ to denote the current MST for hyperedge $s \in S$; the support $G$ is always the union over all $T_s$.
As we compute a spanning tree for each hyperedge, $G$ is a support for $H$ when the algorithm terminates.

\subsubsection{Efficiency}
Implementing $G$ with adjacency lists, we use $O(nk)$ storage as each of the $k$ trees has $O(n)$ edges.
To compute $T_s$, we use Lemma~\ref{lem:iterationtrees} below to conclude that 
there are $O(nk)$ candidate edges, ensuring that Prim's MST algorithm runs in $O(nk + n \log n)$ time.
To see that we can determine the weight without overhead, consider all vertices to be indexed with numbers from $1$ to $n$.
When adding a vertex $u$ to the current tree in Prim's algorithm, we first process the neighbors of $u$ in $G$ (having a weight $0$) and mark that these have been processed in an array using the above mentioned vertex index.
Only then do we process all other vertices (having weight equal to the Euclidean distance) that are not marked and are not in the current tree.
The total algorithm thus takes $O(|\sigma| (nk + n \log n))$ time and $\Theta(nk)$ space.

\begin{restatable}{lemma}{lemiterationtrees}
\label{lem:iterationtrees}
Let $P$ be a point set and $F \subseteq P \times P$.
Consider the MST $T$ on $P$, based on edge weights $0$ for edges in $F$ and the Euclidean distance otherwise.
Then $T$ is a subset of $F$ and the Euclidean MST on $P$.
\end{restatable}

\subsubsection{Properties ($k = 2$)}
The main question that arises is how long a computation sequence $\sigma$ must be such that that the result \emph{stabilizes}, that is, any sequence that extends $\sigma$ gives a support that has the same total length.
We use $G_{\sigma}$ to denote the support resulting from computation sequence $\sigma$. 
Below, we sketch an argument that for $k=2$, we need to only recompute one hyperedge: sequence $\sigma = \langle r,b,r \rangle$ or $\sigma = \langle b,r,b \rangle$ is sufficient to obtain a stable result.
We can compute both sequences and use the result with smallest total edge length. 

\begin{restatable}{lemma}{lemthreeisenough}
\label{lem:three-is-enough}
Let $H = (V, \{r,b\})$ be a hypergraph. All computation sequences $\sigma'$ with $|\sigma'| \geq 4$ have a shorter computation sequence~$\sigma$ with $|\sigma| = 3$ with $G_\sigma = G_{\sigma'}$.
\end{restatable}
\begin{proof}[sketch]
We show that the third computation does not add a new edge with both vertices in $r \cap b$. Hence, the second and fourth computation receive the same input and thus yield the same result.
\qed
\end{proof}

\subsection{Local search}
\label{sec:localsearch}
The algorithm described in Section~\ref{sec:mstiteration} appears to perform well in practice, as shown in Section~\ref{sec:experiments}. However, one may wonder whether other commonly employed heuristic approaches outperform it in the experiments.
We therefore implement a local-search algorithm, specifically, a hill-climbing heuristic.

\subsubsection{Algorithm}
This approach assumes that in the given hypergraph $H = (V,S)$, at least one vertex $v \in V$ occurs in all hyperedges $s \in S$ such that Observation~\ref{obs:star} applies; let $\VA = \bigcap_{s \in S} s \ne \emptyset$. 
We need to initialize our hill climbing approach with a valid (plane), easy to find albeit possibly suboptimal solution.
Following Observation~\ref{obs:star}, we obtain this by first calculating an EMST of all vertices in $\VA$, and subsequently connecting all vertices $v \not\in \VA$ to the nearest $v' \in \VA$.

Afterwards, we iteratively execute rounds until no further improvement is gained. Each round consists of checking for each edge in the support if it can be removed, and if the hyperedges using it can be reconnected by (one or more) other edges that have a shorter total length than the removed edge without causing intersections. This check is nontrivial and done in a brute-force manner, improved by caching and pruning. At the end of each round, the edge replacement that reduces the total edge length most is actually executed. More rounds are evaluated until no single edge replacement reduces the total edge length.

As the initial state is a plane support tree, we can also readily enforce acyclicity, or relax the constraints to allow intersections.

\subsection{Integer linear program}
\label{sec:ilp}
Theorem~\ref{thm:nphard} implies that several variants of computing the shortest plane support are \NP-hard.
Here we briefly sketch how to obtain an \emph{integer linear programs} (ILP) for a hypergraph $H = (V,S)$, allowing us to leverage effective ILP solvers.

We introduce variables $e_{u,v} \in \{0,1\}$, indicating whether edge $uv$ is selected for the support.
This allows us to represent a graph with fixed vertices.
Because the vertex locations are fixed, we can precompute edge lengths $d_{u,v}$ as well as which pairs of edges intersect. This gives the following basic program
\begin{eqnarray*}
\text{minimize}   & \sum_{u, v \in V} d_{u,v} \cdot e_{u,v} \\
\text{subject to} & e_{u,v} + e_{w,x} \leq 1
                  & \text{for all } u,v,w,x \in V  \text{ if edges } uv \text{ and } wx \text{ intersect.}
\end{eqnarray*}

What remains is to ensure that the graph is also a support: we need additional constraints that imply that each hyperedge in $S$ induces a connected subgraph.
To this end, we construct a flow tree for each hyperedge $s$. We pick an arbitrary sink for the hyperedge, $\sigma_s \in s$, that may receive flow, and let the remaining vertices in $s$ generate one unit of flow. To formalize this, we introduce variables $f_{s,u,v} \in \{0, 1, \ldots, |s|-1 \}$ for each $s \in S$ and $u,v \in s$ with $u \neq v$. We now need the following constraints:
(a) the incoming flow at $\sigma_s$ is exactly $|s|-1$;
(b) the outgoing flow at $\sigma_s$ is zero;
(c) except for $\sigma_s$, each vertex in $s$ sends out one unit of flow more than it receives;
(d) flow can be sent only over selected edges.
\begin{eqnarray*}
\text{(a)} & \sum_{u \in s \setminus \{\sigma_s\}} f_{s,u,\sigma_s} = |s| - 1
           & \text{for all } s \in S \\
\text{(b)} & f_{s,\sigma_s,v} = 0
           & \text{for all } s \in S, v \in s \setminus \{ \sigma_s \} \\
\text{(c)} & \sum_{v \in s \setminus \{u\}} (f_{s,u,v} - f_{s,v,u}) = 1
           & \text{for all } s \in S, u \in s\setminus\{ \sigma_s \} \\
\text{(d)} & f_{s,u,v} \leq e_{u,v} \cdot (|s| - 1)
           & \text{for all } s \in S, u,v \in s \text{ with } u \neq v
\end{eqnarray*}

\subsubsection{Variants}
The ILP results in the shortest plane support for $H$. It can easily be modified to give a shortest (plane or unconstrained) support tree as well as to penalize or admit a limited number of intersections. The latter requires additional variables to indicate whether both edges of a crossing pair are used.

\section{Experiments}
\label{sec:experiments}
As discussed above, there are various ways of defining and computing good supports.
In this section we discuss several computational experiments that were performed to gain insight into the trade-offs between the different methods and properties.
In particular, we use two different setups.
First, we exclude optimal but slow algorithms to extensively compare the heuristic algorithms.
Second, we include optimal algorithms to answer questions about the effect of requiring planarity or support trees, and to investigate how well heuristic algorithms approximate the optimal solution, albeit on smaller data sets.

\subsubsection{Algorithms}
We shall study four algorithms under various conditions in these experiments.
In particular, we use \textsc{MSTApproximation} to refer to the simple approximation algorithm of computing a minimum spanning tree for each hyperedge and then taking their union~\cite{coloredspanninggraphs}.
We refer to our heuristic improvement as \textsc{MSTIteration} (Section~\ref{sec:mstiteration}).
Finally, we use \textsc{LocalSearch} to indicate our local search algorithm (Section~\ref{sec:localsearch}) and \textsc{Opt} to denote an exact algorithm for computing optimal solutions.
The latter two allow four different conditions, by requiring a plane support, a support tree, both (i.e., a plane support tree) or neither (unrestricted). We append \textsc{P}, \textsc{T}, \textsc{PT} and \textsc{U} to denote these conditions.

\subsubsection{Data generation}
We generate a random hypergraph $H = (V,S)$ via the procedure described in Appendix~\ref{app:exp_datagen}.
Our method ensures that at least one vertex is an element of all hyperedges (necessary for \textsc{LocalSearch}, see Section~\ref{sec:localsearch}), and that each hyperedge has at least two vertices.
The procedure generates a hypergraph with $n$ vertices, $s$ hyperedges and a degree distribution $d$ according to one of the following scheme:

   \noindent\begin{tabular}{ll}
      ~~\EVEN~~ & All degrees occur equally frequently.\\
      ~~\MID & Degrees are drawn from a normal distribution with a peak on $k/2$.\\
      ~~\LOW & Degrees are drawn from a normal distribution with a peak on $1$.\\
      ~~\HIGH & Degrees are drawn from a normal distribution with a peak on $k$.
    \end{tabular}

\subsection{Experiment 1: comparison of heuristics}

Here we focus on answering the following three questions:
(1) how much does the spanning tree iteration help to reduce the length of the support, compared to computing the minimum spanning trees in isolation;
(2) which heuristic algorithm performs best in terms of support length;
(3) which heuristic algorithm performs best in terms of computation time?

\subsubsection{Setup}
For each combination of $n = 20$, $40$, $60$, $80$, $100$, $k = 2$, $3$, $4$, $5$, $6$, $7$ and $d = \EVEN$, $\MID$, $\LOW$, $\HIGH$, we generate $1000$ random hypergraphs with $n$ vertices and $k$ hyperedges according to degree distribution scheme $d$.
For each hypergraph, we perform six algorithms: \textsc{MSTApproximation} and \textsc{MSTIteration} as well as \textsc{LocalSearch U/T/P/PT}.
This experiment was run on one machine, sequentially in a single thread to also allow for comparison of runtime performance.
The machine was an HP ZBook with an Intel Core i7-6700HQ CPU, 24 GB RAM and running Windows 8.1.

\subsubsection{Results}
We first consider question (1) and compare \textsc{MSTApproximation} and \textsc{MSTIteration}. Since \textsc{MSTIteration} can only improve upon \textsc{MSTApproximation}, we express this as a ratio between 0 and 1.
In Fig.~\ref{fig:mst-ratio} we show the results for $n = 20,60,100$ (Fig.~\ref{fig:app_mst-ratio} in Appendix~\ref{app:exp_heuristics} provides the chart for all cases).
Interestingly, the median gain remains roughly equal as we increase the number of vertices, though the variance becomes lower.
Increasing the number of hyperedges gradually increases the relative gain of \textsc{MSTIteration}.
We also observe a dependency on the degree distribution. In particular, \MID and \EVEN systematically benefit more from iteration than \LOW and \HIGH.
We explain this by observing that in the extreme cases \textsc{MSTApproximation} is optimal: if all vertices have degree 1, then the optimal support is simply the union of all (disjoint) minimum spanning trees; if all vertices have degree $k$, then the optimal support is also simply the minimum spanning tree on the vertices. Difficulties arise when having many vertices that are part of multiple but not all hyperedges.
This corresponds to the \MID and \EVEN schemes.

\begin{figure}[t]
  \centering
  \makebox[\textwidth][c]{\includegraphics[scale=0.48]{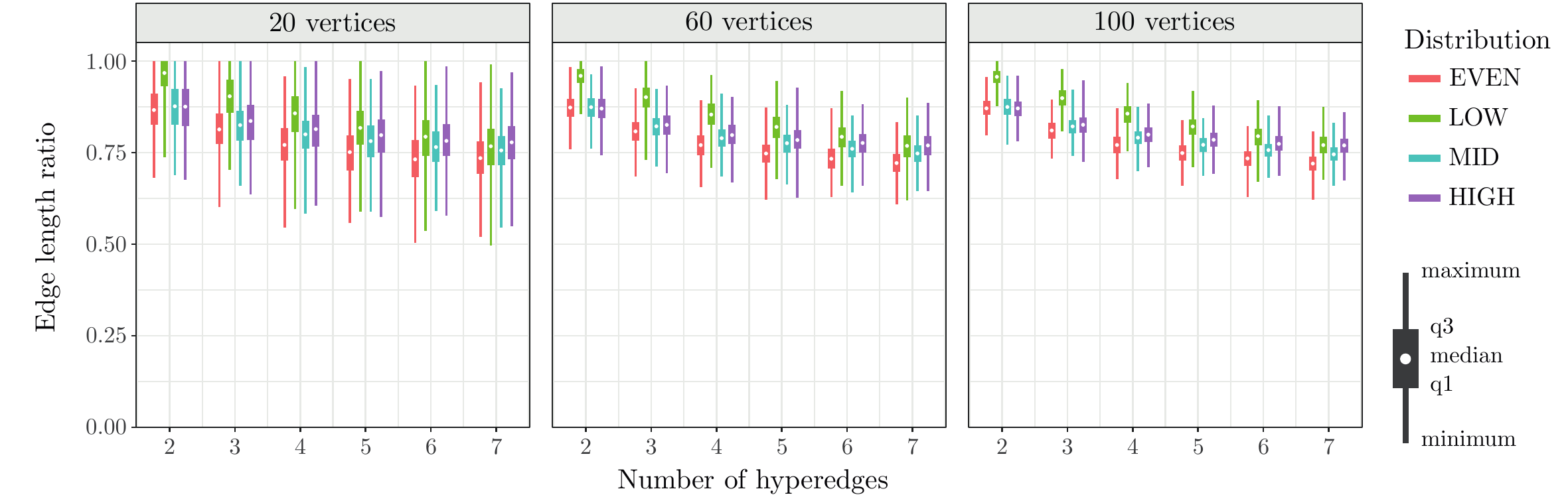}}
  \vspace{-2\baselineskip}
  \caption{Ratio of the support length computed by \textsc{MSTIteration} as a fraction of \textsc{MSTApproximation}. Lower values indicate a higher gain of the iteration method.}
  \label{fig:mst-ratio}
\end{figure}

\begin{figure}[b]
  \centering
  \makebox[\textwidth][c]{\includegraphics[scale=0.48]{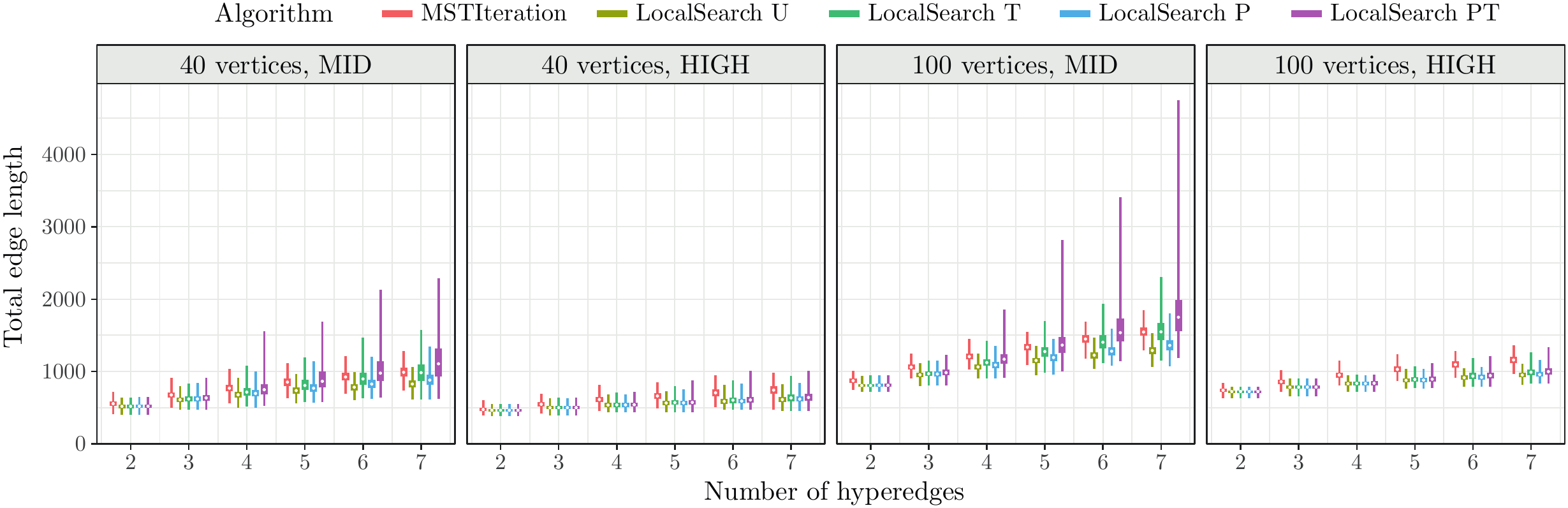}}
  \vspace{-2\baselineskip}
  \caption{Support length computed by the algorithms for varying values of $n$, $k$ and~$d$.}
  \label{fig:length-overview}
\end{figure}

Let us now turn towards question (2), and consider the resulting support length of the \textsc{LocalSearch} algorithm as well.
We omit \textsc{MSTApproximation} from these comparisons, since \textsc{MSTIteration} always performs at least as well.
In Fig.~\ref{fig:length-overview} we show the results for $n = 40$ and $100$ (Fig.~\ref{fig:app_length-overview} in Appendix~\ref{app:exp_heuristics} provides the chart for all cases).
As one may expect, the length increases gradually with more hyperedges, as the support must use more edges to ensure that each hyperedge induces a connected subgraph.
Moreover, we see that \textsc{LocalSearch U} consistently outperforms \textsc{MSTIteration}. To be exact, this is the case in $98.5\%$ of all trials; the average ratio of \textsc{LocalSearch U} to \textsc{MSTIteration} (including those trials in which \textsc{MSTIteration} performs better) is $0.877$, that is, the support length is over $12\%$ shorter on average.
The effect of degree distribution also stands out.
In \LOW and \MID, requiring planarity or a support tree has a large effect on the support length, whereas this is not the case in \EVEN and \HIGH.
To explain this, observe that the minimum spanning tree on vertices that are in many or all hyperedges is planar and likely a part of the computed solution; in the \EVEN and \HIGH cases, there are comparatively many such vertices which can then serve as places to connect the other vertices in the support. In the \LOW and \MID cases, there are only few such vertices and thus the shortest connections that can be used to connect these to such a ``backbone'' structure are likely to intersect other connections.
Though the number of vertices has little effect on \textsc{MSTIteration} and \textsc{LocalSearch U}, this does exacerbate the above problem: more vertices leads to a larger increase in support length when we enforce planarity or a support tree.

\begin{figure}[t]
  \centering
  \makebox[\textwidth][c]{\includegraphics[scale=0.48]{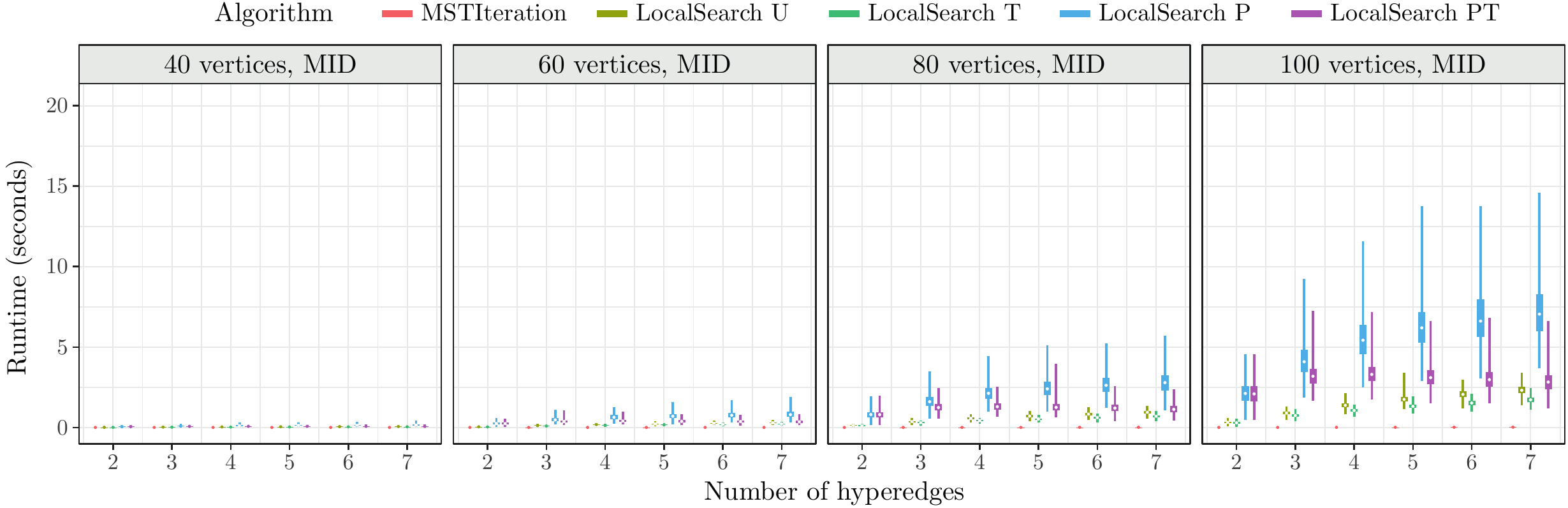}}
  \vspace{-2\baselineskip}
  \caption{Computation time of the various algorithms for varying values of $n$ and $k$.}
  \label{fig:time-overview}
\end{figure}

Finally, we briefly consider question (3) and compare the computation times of the various algorithms (see Fig.~\ref{fig:time-overview}, or Fig.~\ref{fig:app_time-overview} in Appendix~\ref{app:exp_heuristics}).
We see that the number of hyperedges impacts the computation only slightly, whereas the number of vertices has a much stronger effect.
\textsc{MSTIteration} clearly outperforms the \textsc{LocalSearch} variants, running on average $95.11\%$ faster than \textsc{LocalSearch U} over all trials ($98.73\%$ faster on trials with $n = 100$).
Another clear pattern is that requiring planarity with \textsc{LocalSearch} increases the running time significantly ($272.64\%$ slower over all trials, $354.06\%$ on trials with $n = 100$);
the number of steps to arrive at a local minimum is not sufficiently reduced to compensate for the time spent on checking intersections.

\vspace{-0.2\baselineskip}
\subsection{Experiment 2: comparison of optimality}

Here we focus on answering two questions:
(1) how is the support length affected by additionally requiring that the support is a tree and/or is planar;
(2) how well do the heuristic algorithms approximate the optimal solution?

\vspace{-0.1\baselineskip}
\subsubsection{Setup}
For each combination of $n = 10$, $15$, $20$, $k = 2$, $3$ and $d = \LOW$, $\MID$, we generate $1000$ random hypergraphs with $n$ vertices, $k$ hyperedges according to degree distribution scheme $d$.
For each hypergraph, we run the \textsc{LocalSearch U/T/P/PT} and compute an optimal solution \textsc{Opt U/T/P/PT}\footnote{For $n = 10,15$, this is a simple branch and bound algorithm; for $n = 20$ we use the ILP solution, solved with IBM ILOG CPLEX 12.6.3.}.
To obtain a large enough number of trials, these experiments were run on different machines simultaneously and in concurrent threads.
As such, we refrain from analyzing algorithm speed in this experiment.

\vspace{-0.1\baselineskip}
\subsubsection{Failed trials}
In about $3.4\%$ of the CPLEX runs for $n=20$, the computation would run out of memory and therefore not finish successfully.
We ran additional trials to compensate, eventually obtaining 1000 successful trials.
This likely biases the results for $n=20$ towards including only the ``easier'' situations.
Appendix~\ref{app:exp_opt} provides more details including statistics on which cases failed and indicators of the ``difficulty'' of these cases.

\subsubsection{Results}
Let us first compare the optimal solutions according to the four different restrictions.
In Fig.~\ref{fig:opt-overview} we show the results.
For two hyperedges, we see that there is little to no effect of requiring support trees, but a small worst-case effect for requiring plane supports for the \LOW case---the median increases only slightly.
For three hyperedges, we see that the effects become slightly larger. Most noticeable is that enforcing support trees has now a slight effect, even for only a few vertices.
In terms of plane supports, we see a similar pattern as before, that is, that of an increase particularly in the \LOW case, but also some in the \MID case.
Note that the effects for $n=20$ are potentially underestimated.

\begin{figure}[b]
  \centering
  \makebox[\textwidth][c]{\includegraphics[scale=0.48]{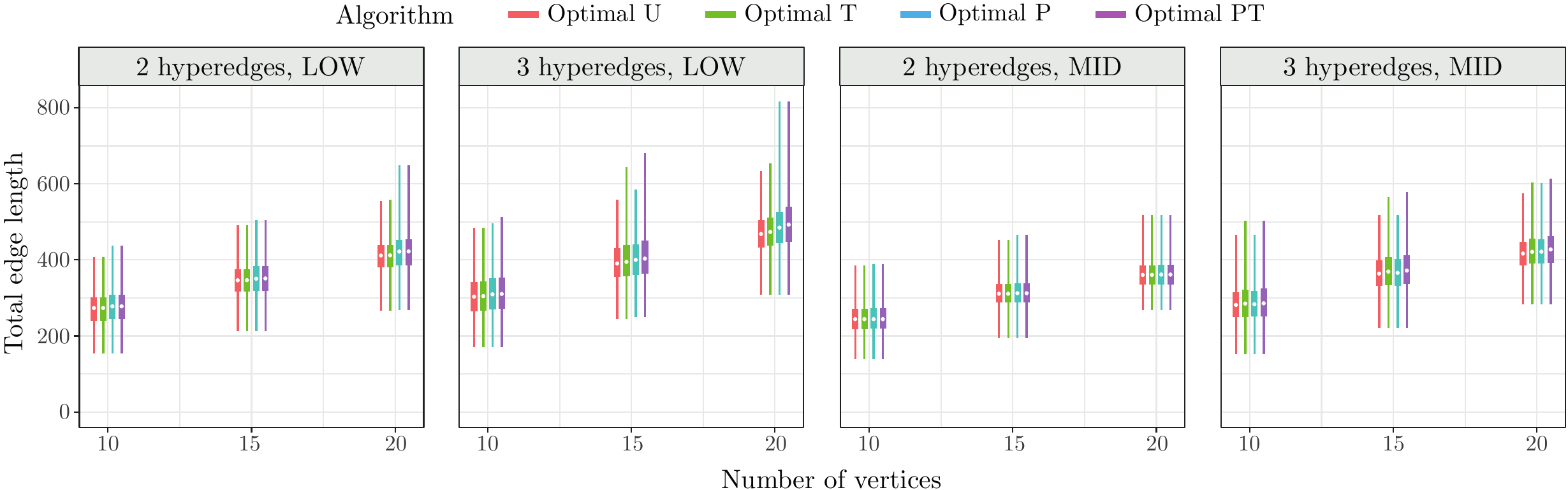}}
  \vspace{-2\baselineskip}
  \caption{Support length achieved by \textsc{Opt} in the four conditions \textsc{U/T/P/PT}.}
  \label{fig:opt-overview}
\end{figure}

\begin{figure}[b]
  \centering
  \makebox[\textwidth][c]{\includegraphics[scale=0.48]{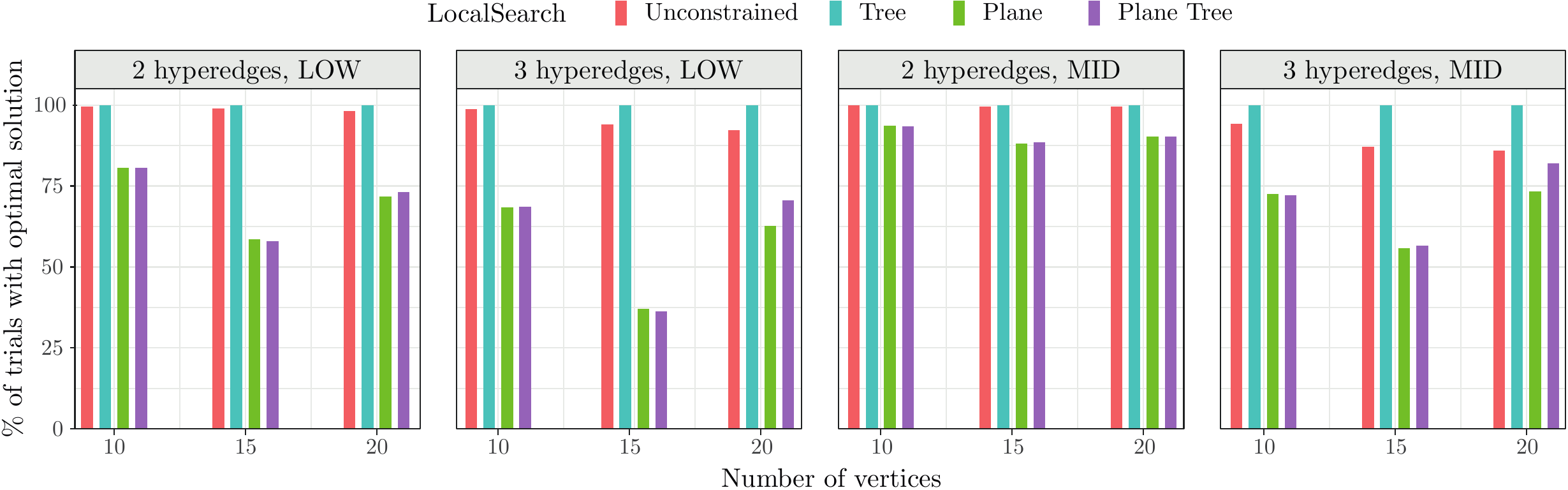}}
  \vspace{-2\baselineskip}
  \caption{Percentage of runs of \textsc{LocalSearch} that achieve the optimal solution. Note that \textsc{LocalSearch T} always achieves optimal results.}
  \label{fig:ls-opt}
\end{figure}

Let us now turn towards how well \textsc{LocalSearch} performs with respect to the optimal solution.
Our results indicate that in a majority of the cases, our heuristic actually achieves optimal results (see Fig.~\ref{fig:ls-opt}).
For $n=10$,$15$ we see a clear decrease of this percentage for plane supports and trees; we attribute the apparent increase at $n=20$ to the failed trials.
To further see how well \textsc{LocalSearch} performs if it fails to achieve optimal results, we look at the ratio between the support length it achieves and the optimal support length. 
In all cases, we observe a ratio of less than $1.61$.
The 90-, 95-, and 99-percentile of this ratio was worst for \textsc{LocalSearch PT}, being $1.05$, $1.09$, and $1.19$, respectively.
Again, we have to keep in mind that the data for $n=20$ likely exclude some more difficult cases and thus the trend in the increasing ratio might extend further for a larger number of vertices.

\section{Conclusion}
\label{sec:conclusion}
Motivated by the \NP-hardness of computing shortest plane supports, we introduced and evaluated two heuristic algorithms for the problem.
Our experiments showed that the heuristic \textsc{LocalSearch} often achieves the optimal solution, and otherwise computes a support that is less than $20\%$ longer than the optimal solution in $99\%$ of the cases.
Moreover, our experiments showed that \textsc{LocalSearch} performs better than \textsc{MSTIteration}, which in turn is a $k$-approximation for $k$ hyperedges.
We can also guarantee that \textsc{LocalSearch} (without restrictions) is a $k$-approximation by initializing it using either \textsc{MSTApproximation} or \textsc{MSTIteration}, though it is not clear whether this change will generally improve the result of \textsc{LocalSearch}.
There is a trade-off between speed and support length, where \textsc{MSTIteration} is better for the former and \textsc{LocalSearch} for the latter.
We also observed that the increase in support length caused by additional requirements, depends both on the number of sets and the number of set memberships per element, but this behavior seems predictable and not to depend on the number of elements.

\subsubsection{Future work}
From the theoretical side, several questions remain open.
For example, can we efficiently decide whether a plane support tree exists? We currently know how to answer this only for two hyperedges (using Observation~\ref{obs:star} and \cite{redbluetrees}).
Furthermore, how many iterations do we need for \textsc{MSTIteration} with more than two hyperedges, to guarantee that the computation stabilizes?

Our experiments indicate that our local search algorithm does not always perform optimally, especially when requiring plane supports.
It is, however, based on simple hill climbing. Can we employ better search techniques such as simulated annealing to efficiently find better solutions?

Finally, we chose to generate random hypergraphs for our experiments, as to not depend on particular properties of (geospatial) configurations that may be inherent to some real-world data sets.
While this reduces the explanatory power with respect to real-world data sets, it provides us with more insight into the structural problem, unbiased by unknown or hidden structures of real-world data.
We leave it to future work to further dive into real-world data sets, to see if similar trends and patterns emerge or more difficult structures arise and to evaluate the impact of the different heuristics on readability.

\subsubsection{Acknowledgments}
This work started at Dagstuhl seminar 17332 ``Scalable Set Visualizations''. The authors would like to thank Nathalie Henry Riche for providing the data for Fig.~\ref{fig:realdata}. TC was supported by the Netherlands Organisation for Scientific Research (NWO, 314.99.117). MvG received funding from the European Union's Seventh Framework Programme (FP7/2007-2013) under ERC grant agreement n$^{\text{o}}$~319209 (project NEXUS 1492) and the German Research Foundation (DFG) within project B02 of SFB/Transregio~161. WM was partially supported by the Netherlands eScience Centre (NLeSC, 027.015.G02).

\clearpage
\appendix

\section{Appendix: Omitted proofs}
\label{app:proofs}
\lemapprox*
\begin{proof}
The hypergraph family is illustrated in Fig.~\ref{fig:approx} (on page~\pageref{fig:approx}).

The set $\VA = \{u,v,w\}$ consists of three vertices whose EMST $T$ has length $\ell + 1$ and is indicated by the black edges in Fig.~\ref{fig:approx}(a).
The remaining vertices in $V \setminus \VA$ are indicated in red and blue (indicating membership of $r$ and $b$) and placed inside a disk of radius $\varepsilon$ just left of the midpoint of edge $uv$.
The vertices alternate in colors from left to right and form two mirrored convex chains.

Since edge $uv$ of $T$ splits the vertices in $V \setminus \VA$ and by their placement on convex chains, the shortest extension of $T$ into a plane support tree is to connect every vertex to $u$ (Fig.~\ref{fig:approx}(a)).
This yields a total length of the support tree of $\Theta(n) \cdot \ell$.
If, however, $\VA$ is connected by a slightly longer tree, the remaining vertices in $V \setminus \VA$ can be joined by two comb-shaped structures as shown in Fig.~\ref{fig:approx}(b).
The resulting plane support tree has length of $\Theta(1) \cdot \ell$.
\qed
\end{proof}

\thmnphard*
\begin{proof}
We first show the reduction for the case that $r \subseteq b$.
We use a reduction from planar monotone 3-SAT \cite{l-pftu-82}. 
Here, we are given a 3-CNF formula $\phi$ with $n$ variables $v_1, \ldots, v_n$ and $m$ clauses $c_1, \ldots, c_m$ such that every clause either has three positive literals or three negative literals. Moreover, we are given an embedding of~$\phi$ as a graph, with rectangular vertices for variables on a horizontal line, and clauses as rectangles above or below the line (depending on whether the clause is positive or negative). Vertical edges connect clauses to the variables of their literals.

We must construct a hypergraph $H = (V,\{r,b\})$ such that $r \subseteq b$. In the remainder, we assign vertices to either $r$ (red) or $b$ (blue), understanding that any red vertex is also in $b$.

First, we place $3(n+1)$ red vertices using coordinates $(3 i \cdot (m+1), y)$ for integers $i \in [0,n]$ and integers $y \in [-1,1]$. Furthermore, we place $n \cdot (3m + 2)$ blue vertices using coordinates $(3i(m+1) + j, 0)$ for integers $i \in [0,n-1]$ and $j \in [1,3m+2]$.

We now place additional blue vertices for each clause $c_a$. We assume that this clause has positive literals for variable $v_i$, $v_j$, and $v_k$; the construction for clauses with negative literals is symmetric, using negative $y$-coordinates instead. First, we place $3a+1$ blue vertices from $(3(i-1)(m+1) + 3p,2)$ to $(3(i-1)(m+1)+3p,2+3a)$ at unit distance, to represent the incidence from $c_a$ to variable $v_i$, using the given embedding to determine that $c_a$ is the $p$th clause incident from above to $v_i$.
Analogously, we place the blue vertices for $v_j$ and $v_k$.
Now, we place further blue vertices at unit distance with $y$-coordinate $2+3a$ from the leftmost to the rightmost top vertex we just placed.
The result is given in Fig.~\ref{fig:nphard}.

\begin{figure}[tb]
  \centering
  \includegraphics[page = 2]{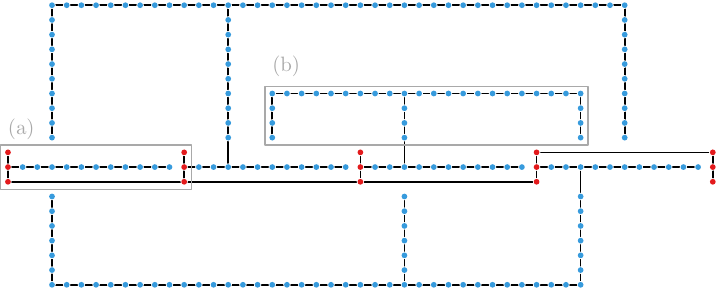}
  \caption{Construction for $\phi = (v_2 \vee v_3 \vee v_4) \wedge (\overline{v_1} \vee \overline{v_3} \vee \overline{v_4}) \wedge (v_1 \vee v_2 \vee v_4)$. Vertices in $r$ and $b$ are red, vertices in $b$ are blue. A plane support tree with length at most $L$ is given in black lines. (a)~Representation of variable $v_1$; the solution sets $v_1$ to true. (b) Representation of the first clause.}
  \label{fig:nphard}
\end{figure}

One clause requires at most $3(3m+1)$ vertices for the variable incidence and less than $3n\cdot(m+1)$ for the horizontal line connecting these. We can now readily measure the length of the minimum spanning tree on the blue vertices of one clause. We use $L_a$ to denote this length; note that $L_a$ is an integer at most $3(3m+1) + 3n\cdot(m+1)$.

The value of $L$ that we select is $2(n+1) + 3n \cdot (m+1) + n(3m+2)  + 2m + \sum_{a \in [1,m]} L_a$.

This finalizes the construction. It is polynomial since we placed $3(n+1)$ red vertices and $n \cdot (3m-2)$ blue vertices for the variables and at most $m \cdot (3(3m+1) + 3n\cdot(m+1))$ for the clauses: this is $O(nm^2)$ vertices.
Moreover, we claim that our constructed hypergraph admits a plane support tree of length at most $L$, if and only if $\phi$ is satisfiable.

Assume we have a plane support tree of length at most $L$. First, we observe that all points in $r$ must be connected: the minimal way of doing so connects the three vertices with the same $x$-coordinate and uses one horizontal line to connect one triplet to the next.
This has exactly length $2(n+1) + 3n \cdot (m+1)$, corresponding to the first two terms defining $L$.
The minimal way of connecting the lines inside the variables to the red tree takes length $n(3m+2)$ in total: this is the third term defining $L$. Finally, to connect the clause vertices, we need length at least $L_a$ per clause, the last term of $L$. We note that any solution must use these constructions on the blue vertices, since all vertices are at unit distance; other blue vertices are at distance at least $2$.
However, the support tree is connected: thus it must still have connections from each gadget to either a red vertex or a blue vertex of a variable.
The budget we have for this is $2m$ in total.
Since each clause needs a connection of length at least $2$, all clauses use exactly length $2$.
The only vertices within distance $2$ of a clause are the three blue vertices of the variables with $y$-coordinate zero (one of each literal of the clause).
Thus, each clause must have exactly one length-$2$ edge to one of these variable vertices. Since the support tree is plane, this cannot cross the horizontal links used to connect the red vertices. We can now readily obtain a satisfying assignment for $\phi$, by looking at which of the two horizontal lines is used to connect the red vertices: if the one at the top is used, that variable is set to false; it is set to true otherwise.

To prove the converse, assume that we have a satisfying assignment. Using the same reasoning as above, we can construct the plane support tree by picking the connecting horizontal lines for the red vertices according to the satisfying assignment: this readily leads us to conclude that we can connect each clause using a length-$2$ connection that does not intersect the horizontal lines for the red vertices.

Finally, let us consider the case that $r$ and $b$ are disjoint. The reduction can easily be amended to work for this case: the red vertices are only in $r$ rather than $r$ and $b$.
This then needs slightly more spacing such that we can add a few extra blue vertices that can be used to connect all the blue vertices of the variables into a single component using only length-$1$ edges.
\qed
\end{proof}

\lemiterationtrees*
\begin{proof}
Let $T'$ denote the Euclidean MST on $P$.
Assume that MST $T$ has some edge $e$ that is neither in $F$ nor in~$T'$.
Since $T$ is a tree, removing $e$ from it partitions the tree into two connected components.
By definition, $T'$ contains an edge $e'$ that connects the two components and by assumption~$e' \neq e$.
Since $T'$ is the Euclidean MST\footnote{This assumes either unique distances between all pairs of vertices, or a deterministic way of choosing which edge goes in the MST when multiple have the same minimum weight. The latter can easily be implemented in practice and is as such a reasonable assumption.}, we know that $\| e' \| < \| e \|$, where $\|\cdot\|$ denotes the Euclidean length.
Since $e$ is not in $F$, the weight it contributes to $T$ is $\|e\|$ and thus we can find a shorter spanning tree $T^*$, by replacing $e$ with $e'$ in $T$.
This contradicts that $T$ is the MST, thus proving the lemma.
\qed
\end{proof}

\begin{observation}\label{obs:alternate}
A computation sequence featuring two consecutive occurrences of the same hyperedge achieves the same result as the computation sequence in which these consecutive occurrences have been replaced by a single occurrence.
Hence, any computation sequence (that is not equivalent to some shorter sequence) consists of alternating $r$'s and $b$'s.
\end{observation}

\lemthreeisenough*
\begin{proof}
By Observation~\ref{obs:alternate}, consider $\sigma'$ to start either \textit{(i)} with $\langle r,b,r,b,\ldots \rangle$ or \textit{(ii)} with $\langle b,r,b,r,\ldots \rangle$. We will show that the subsequence $\sigma$ consisting of the first three hyperedges of $\sigma'$ achieves the same support as $\sigma'$.

Consider all edges $(v_i, v_j) \in V\times V$. There are four cases:
\begin{itemize}
 \item If both $v_i$ and $v_j$ are in both $r$ and $b$, let the edge be in a set $\Pu$ of purple edges.
 \item Else, if $v_i$ and $v_j$ are both in $r$, let the edge be in a set $R$ of red edges.
 \item Else, if $v_i$ and $v_j$ are both in $b$, let the edge be in a set $B$ of blue edges.
 \item Else, the edge will never be a part of a support as the vertices do not share a color.
\end{itemize}
Without loss of generality we consider case \textit{(i)}. Let the support constructed after step $i$ of $\sigma$ be called $G_i$, so that we have $G_1$, $G_2$ and $G_3$. We show that $\Pu(G_2) = \Pu(G_3)$, where $\Pu(G)$ denotes taking the subset of edges of $G$ that are in $\Pu$.
\begin{description}
 \item[$\Pu(G_3) \subseteq \Pu(G_2)$.] Let $e_p \in \Pu(G_3)$. For a contradiction, assume $e_p \not\in \Pu(G_2)$. As edges in $\Pu$ are never removed from the support once they are added -- they have weight $0$, after all --, we have $e_p \not\in \Pu(G_1)$ either. As $G_1$ is the Euclidean MST of $r$,  by the cut property of MSTs there is another edge $e \in R \cup \Pu$ shorter than $e_p$ in the cut induced by $e_p$ that must be a part of the MST instead.\footnote{This requires the same assumption of unique distances or determinism as Lemma~\ref{lem:iterationtrees}.} When constructing $G_3$, again $e$ will be chosen over $e_p$, and thus $e_p \not\in \Pu(G_3)$. \lightning
 \item[$\Pu(G_2) \subseteq \Pu(G_3)$.] Let $e_p \in \Pu(G_2)$. We already established that edges in $\Pu$ are never removed from the support once they are added, hence $e_p \in \Pu(G_3)$.
\end{description}
Next, we show that $G_4 = G_3$, i.e., $G_{\sigma'} = G_{\sigma}$.
\begin{description}
 \item[$G_3 \subseteq G_4$.] Take an edge $e \in G_3$. For a contradiction, assume $e \not\in G_4$. As edges in $\Pu$ are not removed and edges in $R$ remain untouched, $e \in B$. As $e \not\in G_4$ and the fourth step calculates $\text{MST}(b)$, the cut property tells us that some other edge $e' \in B \cup \Pu$ is shorter and in $\text{MST}(b)$ instead. But then $e'$ would have been added in $G_2$ and hence $e \not\in G_3$. 
 \lightning
 \item[$G_4 \subseteq G_3$.] Take an edge $e \in G_4$. For a contradiction, assume $e \not\in G_3$. This means $e \not\in R$, as such edges cannot be added when computing $\text{MST}(b)$. Edges in $\Pu$ are never removed, thus $e \not\in G_2$. The second step of $\sigma$ computed $\text{MST}(b)$, hence by the cut property there must be another edge $e'$, shorter than $e$, part of $\text{MST}(b)$ instead. Indeed, this implies $e \not\in \text{MST}(b)$. However, as $G_4$ is computing an MST for $b$ and we assumed $e \in G_4$, $e \in \text{MST}(b)$. \lightning \qed
\end{description}
\end{proof}

\section{Appendix: Experimental results}
\label{app:experiments}
This appendix provides additional details regarding the experiments of Section~\ref{sec:experiments}.
Upon acceptance, we intend to make a version available on ArXiv to provide these additional details.

\subsection{Data generation}
\label{app:exp_datagen}

We generate a random hypergraph $H = (V,S)$ via to the procedure below. We use $n = |V|$ and $k = |S|$ to denote the desired number of vertices and hyperedges respectively.
\begin{enumerate}
  \item \label{alg:degrees} Initialize an array $D[1 \ldots k]$ such that $\sum_{i = 1}^k D[i] = n$, in which $D[i]$ indicates that we wish to generate $D[i]$ vertices of degree $i$. To this end, we define four schemes, where we always restrict the degrees to be between $1$ and $k$.
   \begin{description}
      \item[\EVEN] All degrees occur equally frequently. If $n \!\!\mod k \neq 0$, then degrees one through $n \!\!\mod k$ occur once more than the others.
      \item[\MID] We generate $n$ random degrees using a normal distribution. We draw a random value $g$ from $\mathcal{N}(0.5, 2/9)$ and map this to degree $1 + \lfloor k g \rfloor$. The distribution of degrees is expected to look like a Gaussian curve with its peak on $k/2$.
      \item[\LOW] Similar to the \MID scheme, we draw a random value $g$ from $\mathcal{N}(0, 2/5)$ and map this to degree $1 + \lfloor k |g| \rfloor$. The distribution of degrees is expected to look like a Gaussian curve with its peak on $1$.
      \item[\HIGH] Similar to the \MID scheme, we draw a random value $g$ from $\mathcal{N}(0, 2/5)$ and map this to degree $k - \lfloor k|g| \rfloor$. The distribution of degrees is expected to look like a Gaussian curve with its peak on $k$.
    \end{description}
  \item \label{alg:ensure1all} If $D[k] = 0$, decrease the maximal degree $i$ for which $D[i] > 0$ by one and set $D[k]$ to one.
  \item \label{alg:ensure2k} While $\sum_{i = 1}^k i \cdot D[i] < 2k$, decrease the minimal degree $i$ for which $D[i] > 0$ by one and increase $D[i+1]$ by one.
  \item \label{alg:genvertices} While $\sum D[i] > 0$, let $i$ be a degree such that $D[i] > 0$, chosen uniformly at random. Generate a vertex $v$ with a uniformly random position in a square of width 100 and add it to $V$. Pick $i$ hyperedges uniformly at random from those hyperedges that have less than two vertices; if there are no such hyperedges left, pick from all hyperedges instead. Decrease $D[i]$ by one.
\end{enumerate}
To explain the four steps in this algorithm, we treat them in reverse order.
\begin{enumerate}
  \nextitemref{alg:genvertices}
  \item We generate all desired $n$ vertices and assign them to hyperedges.
    We first pick from those hyperedges that have less than two vertices, to ensure that each hyperedge contains at least two vertices.
    This ensures that all hyperedges have influence on the support.
    We pick a random degree, to avoid biasing small hyperedges towards low degree or high degree vertices.

  \nextitemref{alg:ensure2k}
  \item We ensure that the sum over all degrees (over all nodes) is at least $2k$.
    We need this lower bound on the sum of degrees, to ensure that we are able to pick at least two vertices for every hyperedge.

  \nextitemref{alg:ensure1all}
  \item We ensure that there is at least one vertex that occurs in all hyperedges; this step is optional but necessary to ensure that our local search algorithm can be initialized. It guarantees that a planar solution exists, see Section~\ref{sec:localsearch}.

  \nextitemref{alg:degrees}
  \item We decide on the distribution over the degrees.
    That is, how many vertices shall we have of degree $i$? This can be done according to various schemes. The four schemes used in this paper are described in the main text.
\end{enumerate}

\clearpage
\subsection{Experiment 1}
\label{app:exp_heuristics}

\begin{figure}[htp]
  \centering
  \makebox[\textwidth][c]{\includegraphics[scale=0.48, angle=90]{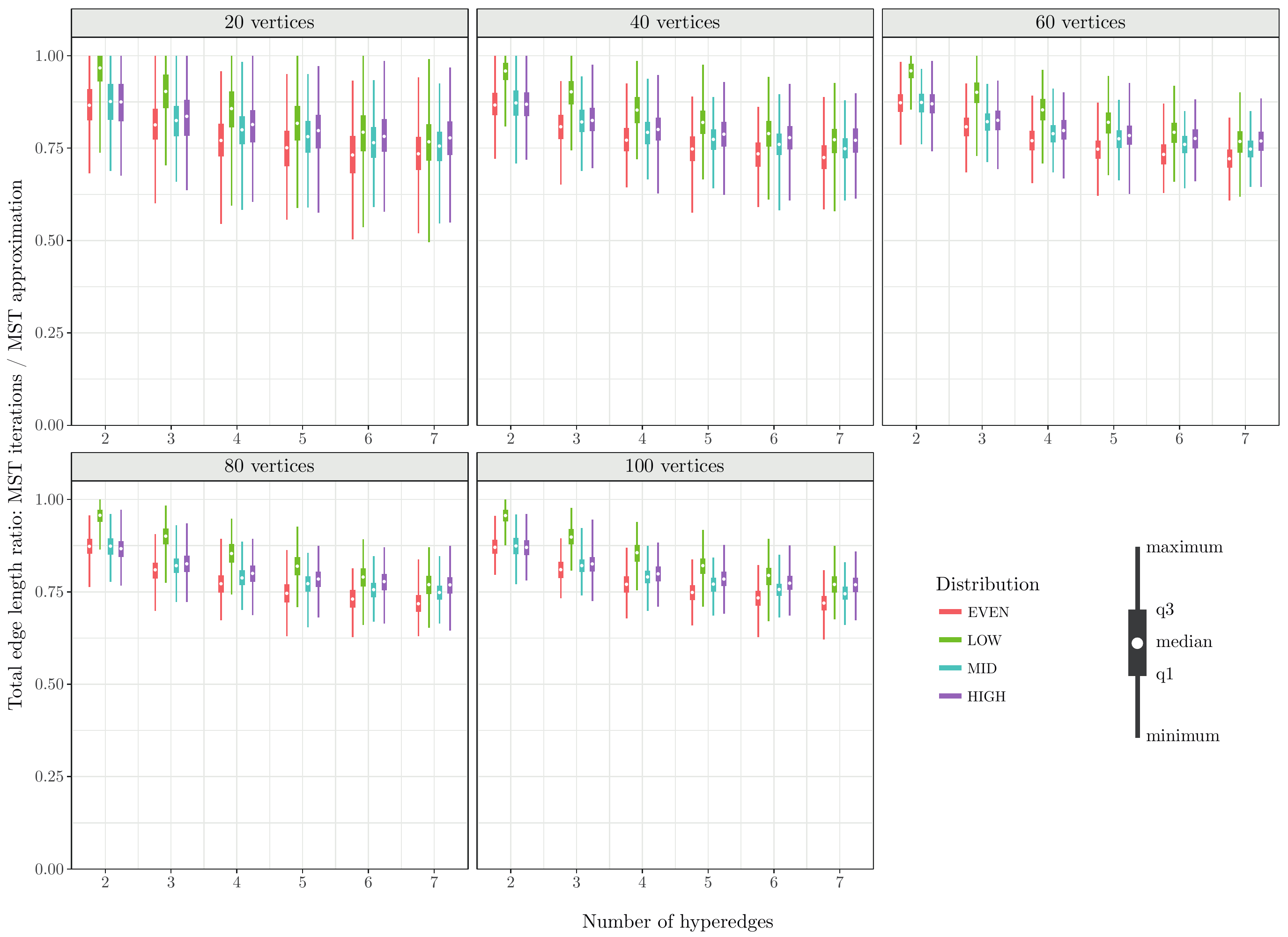}}
  \caption{Ratio of the support length computed by \textsc{MSTIteration} as a fraction of \textsc{MSTApproximation}. Lower value indicate a higher gain of the iteration method.}
  \label{fig:app_mst-ratio}
\end{figure}

\begin{figure}[htp]
  \centering
  \makebox[\textwidth][c]{\includegraphics[scale=0.48, angle=90]{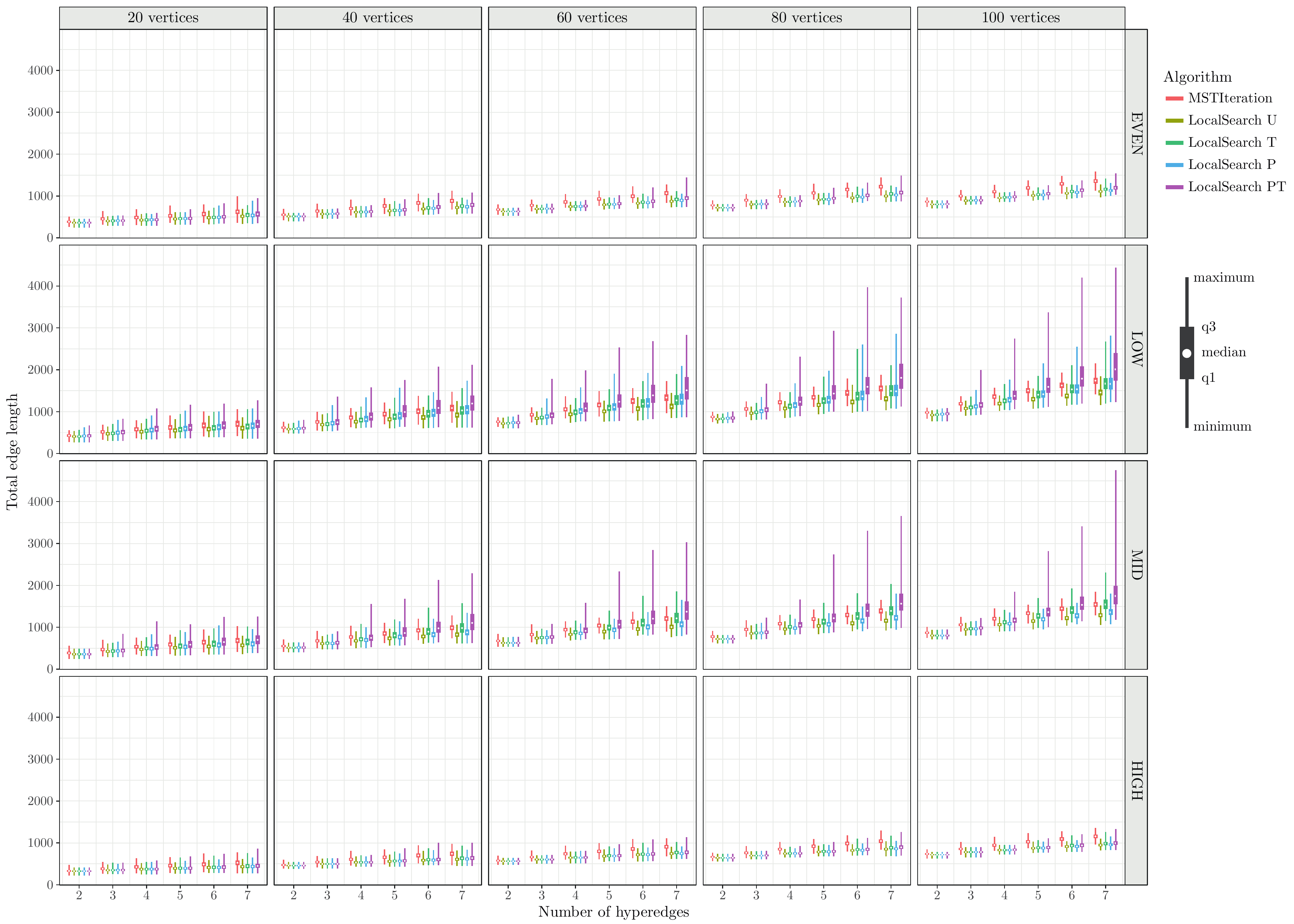}}
  \caption{Support length computed by the various algorithms for varying values of $n$, $k$ and~$d$.}
  \label{fig:app_length-overview}
\end{figure}

\begin{figure}[htp]
  \centering
  \makebox[\textwidth][c]{\includegraphics[scale=0.48, angle=90]{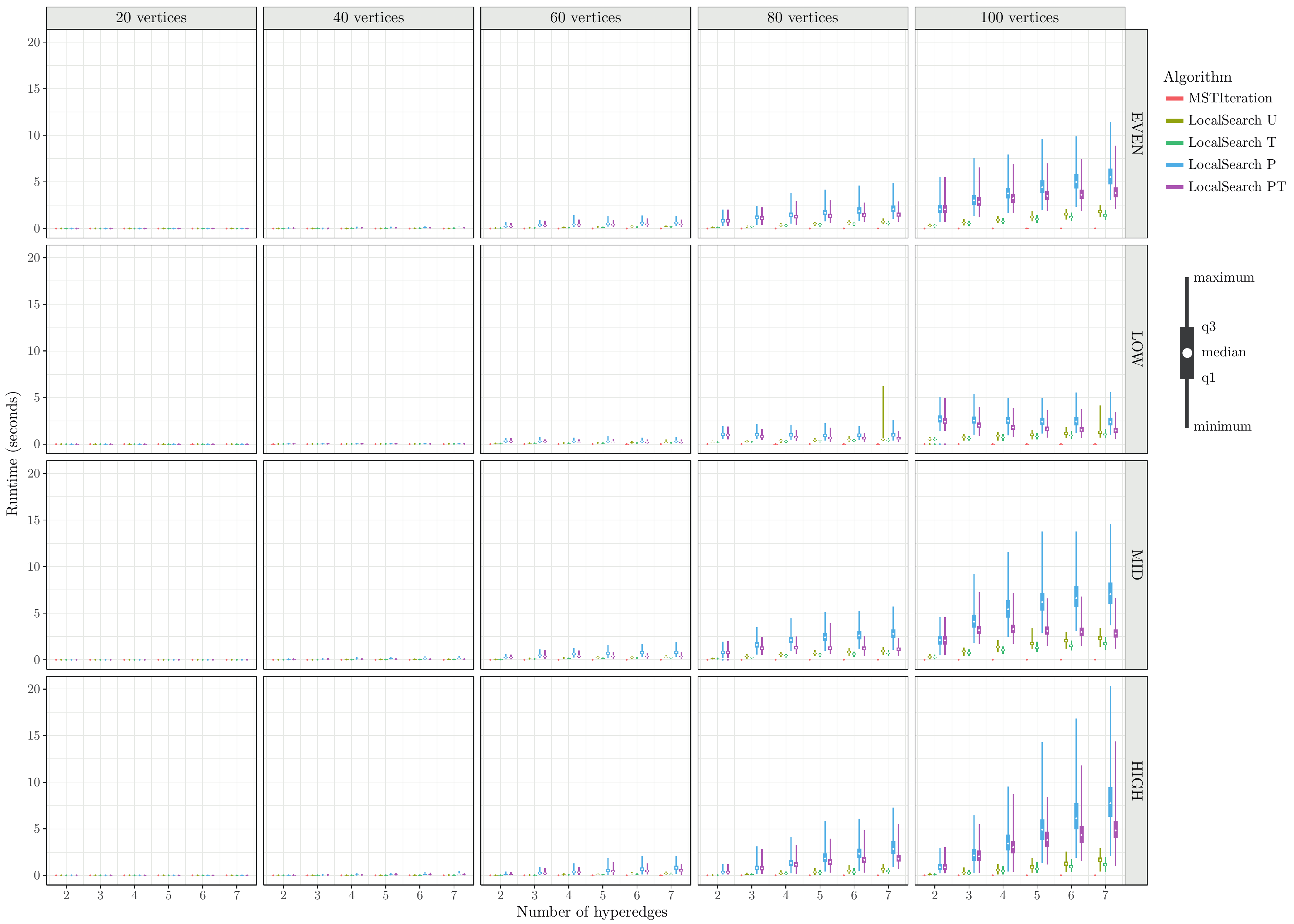}}
  \caption{Computation time of the various algorithms for varying values of $n$, $k$ and~$d$.}
  \label{fig:app_time-overview}
\end{figure}

\clearpage
\subsection{Experiment 2}
\label{app:exp_opt}

\subsubsection{Failed trials}
CPLEX was allocated 24GB of RAM and 64GB of file storage.
Nonetheless, the CPLEX computation would run out of memory and therefore not finish successfully for some cases with $n=20$.
We have therefore ran $1730$ trials for each of the four conditions ($k \times d$) with four settings for \textsc{Opt}; $941$ runs out of these $27,680$ runs failed.
This is shown in Table~\ref{tab:fails}.
We filtered out erroneous trials, leaving 1138 trials, 1000 of which were used for the analysis of the results to match the cases for $n=10,15$.
This may bias the results towards only including the easier cases on which CPLEX was successful; this should be taken into consideration for the upcoming results discussion.
To localize and quantify this bias, we counted which conditions failed and, for each condition, measured the average length of the \textsc{LocalSearch} results in the successful and failed trials (see Table~\ref{tab:fails}).
We note that the tree and plane tree cases are impacted most significantly.
We also see that the ratio is mostly well above one, suggesting that indeed the more difficult cases have now been excluded from the analysis.

\begin{table}[hb]
  \centering
  \caption{Number of failed trials for $n = 20$ per condition. Ratio indicates the average length of \textsc{LocalSearch} on failed trials, divided by the average length of \textsc{LocalSearch} on successful trials.}
  \label{tab:fails}
  \begin{tabular}{*{6}{rr}}
    \toprule

    \multicolumn{2}{c}{} & \multicolumn{2}{c}{\textsc{Opt U}} & \multicolumn{2}{c}{\textsc{Opt T}} &
    \multicolumn{2}{c}{\textsc{Opt P}} & \multicolumn{2}{c}{\textsc{Opt PT}} & \multicolumn{2}{c}{all} \\

    \multicolumn{1}{r}{$k$} & \multicolumn{1}{c}{$d$} & \multicolumn{1}{c}{count} &
    \multicolumn{1}{c}{ratio} & \multicolumn{1}{c}{count} & \multicolumn{1}{c}{ratio} &
    \multicolumn{1}{c}{count} & \multicolumn{1}{c}{ratio} & \multicolumn{1}{c}{count} &
    \multicolumn{1}{c}{ratio} & \multicolumn{1}{c}{count} & \multicolumn{1}{c}{ratio} \\

    \cmidrule(r){1-2}\cmidrule(lr){3-4}\cmidrule(lr){5-6}\cmidrule(lr){7-8}\cmidrule(lr){9-10}\cmidrule(l){11-12}

    2 & \multicolumn{1}{c}{\LOW} & 2 & 1.11 & 7 & 1.26 & 2 & 1.12 & 15 & 1.23 & 26 & 1.22 \\
      & \multicolumn{1}{c}{\MID} & 7 & 1.10 & 3 & 1.26 & 7 & 1.10 & 5 & 1.22 & 22 & 1.15 \\
    3 & \multicolumn{1}{c}{\LOW} & 0 &  & 61 & 1.20 & 0 &  & 264 & 1.26 & 325 & 1.25 \\
      & \multicolumn{1}{c}{\MID} & 18 & 1.13 & 169 & 1.18 & 20 & 1.11 & 361 & 1.23 & 568 & 1.20 \\

    \cmidrule(r){1-2}\cmidrule(lr){3-4}\cmidrule(lr){5-6}\cmidrule(lr){7-8}\cmidrule(lr){9-10}\cmidrule(l){11-12}

    \multicolumn{2}{c}{all} & 27 & 1.09 & 240 & 1.24 & 29 & 1.07 & 645 & 1.33 & 941 & 1.29 \\

    \bottomrule
  \end{tabular}
\end{table}

\end{document}